\documentclass[journal]{IEEEtran}
\usepackage{epsfig,graphicx,rotating}
\usepackage{color}
\usepackage{tabularx}
\usepackage{multicol,multirow,longtable}
\usepackage{amssymb,amsmath,amsthm}
\usepackage{epstopdf}
\usepackage{placeins}
\usepackage{xcolor}
\usepackage{bbm}
\usepackage{amsfonts}
\usepackage{subcaption}
\usepackage[T1]{fontenc}
\usepackage[nolist,nohyperlinks]{acronym}
\usepackage{cite}
\usepackage{mathrsfs}
\usepackage{bbm}
\usepackage{dsfont}
\usepackage{flushend}
\usepackage{symbol_macros_Aritra}

\newenvironment{remark}[1][Remark]{\begin{trivlist}
\item[\hskip \labelsep {\bfseries #1}]}{\end{trivlist}}
\newtheorem{theorem}{Theorem}
\newtheorem{lemma} {Lemma}

\begin{document}
%\font\myfont=cmr12 at 20pt           %%%%%%%% Small Font Title
%\title{{\myfont Spectral and Energy Efficiency Analysis of 3-D Dense Cellular Networks with Realistic Propagation Conditions}}
%\title{{ Spectral and Energy Efficiency Analysis of 3-D Dense Cellular Networks with Realistic Propagation Conditions}}
\font\myfont=cmr12 at 20pt
\title{{ Coverage Analysis of 3-D Dense Cellular Networks with Realistic Propagation Conditions}}
\author{Aritra Chatterjee,~\IEEEmembership{Student Member,~IEEE,}   
         and Suvra Sekhar Das,~\IEEEmembership{Member,~IEEE}, %\newline
\IEEEauthorblockA{Indian Institute of Technology, Kharagpur, India}}
\maketitle 
\begin{acronym}
 \acro{PPP}{Poisson Point Process}
 \acro{SPPP}{Spatial Poisson Point Process}
 \acro{UDN}{Ultra Dense Networks}
 \acro{AP}{Access Points}
 \acro{UE}{User Equipments}
 \acro{PL}{Path Loss}
 \acro{LOS}{Line of Sight}
 \acro{NLOS}{Non-Line of Sight}
 \acro{RV}{random variable}
 \acro{PDF}{probability density function}
 \acro{CCDF}{complimentary cumulative distribution function}
 \acro{CDF}{cumulative distribution function}
 \acro{NSE}{network spectral efficiency}
 \acro{NEE}{network energy efficiency}
 \acro{RAN}{Radio Access Network}
 \acro{SCN}{Small Cell Networks}
 \acro{PGFL}{Probability Generating Functional}
 \acro{SSF}{small scale fading}
 \acro{LSF}{large scale fading}
 \end{acronym}

\begin{abstract}
In recent times, use of stochastic geometry has become a popular and important tool for performance analysis of next generation dense small cell wireless networks. Usually such networks are modeled using 2 dimensional spatial Poisson point processes (SPPP). Moreover, the distinctive effects of line-of-sight (LOS) and non line-of-sight (NLOS) propagation are also not explicitly taken into account in such analysis. The aim of the current work is to bridge this gap by modeling the access point (AP) and user equipment (UE) locations by 3 dimensional SPPP and considering the realistic LOS/NLOS channel models (path loss and small scale fading) as reported in existing standards. The effect of UE density on downlink coverage probability has also been investigated. In this process, the probabilistic activity of APs has been analytically modeled as a function of AP and UE densities. The derived upper bound of coverage probability is found to be numerically simple as well as extremely tight in nature, and thus can be used as a close approximate of the same.  
\end{abstract} 
\section{Introduction}
Densification of \ac{AP} and decentralization of cellular network infrastructure has already been identified as potential weapon to cater the ever-increasing customer demands in future generation wireless \ac{RAN} \cite{andrews_2014}. In recent times, the locations of \ac{AP}s/\ac{UE} in such heterogeneous \ac{SCN} have been modeled as random \ac{SPPP} \cite{Elsawy_2013}. The performances of such networks has been analytically evaluated using tools from stochastic geometry.   \\
\indent In one of the earliest works, an information theoretic approach, named ``Wyner Model'' has been used where the \ac{AP} locations are modeled as 1-dimensional (1-D) point process \cite{Shamai_1997ab}. Although the method is analytically tractable, it is significantly oversimplified due to the 1-D \ac{AP} location modeling and assuming unit power received from each \ac{AP} to the user. In the seminal work \cite{Andrews_2011}, the \ac{AP} locations are modeled as 2-D \ac{SPPP}. The downlink coverage probability and ergodic rate experienced by a typical \ac{UE} have been evaluated using important tools of stochastic geometry like \ac{PGFL}. The coverage probability expression yields to extremely simple closed-form for specific special case like no-noise and path loss exponent equaling 4. However, the channel model used in this work is simple single slope \ac{PL} with exponential distributed \ac{SSF}. \\
\indent The 2-D modeling of \ac{AP} locations is based on the assumption that the inter-\ac{AP} distances and distances from the \ac{UE} and neighboring \ac{AP}s are much larger than the height at which the \ac{AP}s/\ac{UE}s are mounted. This assumption remains valid for moderately dense networks, but not in scenarios like dense urban or dense indoor hotspot and especially in case of \ac{UDN}. Consequently, in recently times, several works have been presented where \ac{SPPP}s of more than 2 dimensions are used \cite{Gupta_2015, Pan_2015, Omri_2016}. In \cite{Gupta_2015}, the \ac{SPPP} is considered to be defined in an upper hemispheric space ($\text{3D}^+$). Dual slope path loss model has been considered to capture distinctive propagation loss effects of \ac{LOS} and \ac{NLOS} cases. However, the \ac{SSF} has been modeled as exponentially distributed for both the cases, making the analysis simplified but not in line with practical scenario. In \cite{Pan_2015}, the downlink performances of a single tier \ac{SCN} modeled by 3-D \ac{SPPP} has been presented, whereas in \cite{Omri_2016}, the analysis is extended to multiple tiers of \ac{AP}s. In both the works, single slope \ac{PL} along with only exponential \ac{SSF} have been considered, which limits their applicability to understand the performance of dense urban/indoor scenarios where both \ac{LOS}/\ac{NLOS} propagation are prevalent. The notion of differentiated propagation loss in \ac{LOS} and \ac{NLOS} has been partially taken into account in \cite{Galiotto_2014, Galiotto_2015, Galiotto_2017} in a 2-D \ac{SPPP} modeled \ac{SCN}. Whereas \cite{Galiotto_2014} presented a simulation based study, effect of dual slope path loss corresponding to \ac{LOS} and \ac{NLOS} propagation has been captured analytically in \cite{Galiotto_2015, Galiotto_2017}. However, the \ac{SSF} in both kind of propagation is considered to be exponentially distributed, making the analysis oversimplified and far from realistic phenomenon. In recently published \cite{Atzeni_2018}, effect of \ac{LOS}/\ac{NLOS} propagation has been captured for \ac{SCN} modeled by 2-D \ac{SPPP} with a fixed \ac{AP} height. However, such analysis does not account for randomness in elevation domain and thus cannot be used to understand the performance ``multi-floor'' dense indoor scenarios. Moreover, the derived expression of coverage probability is numerically complex is nature and thus is not suitable for further analysis. 
So, to summarize, it can be observed that the existing literature either do not consider the random location modeling of network nodes in 3-D, or even if they do, do not properly take into consideration the realistic propagation losses in \ac{LOS} and \ac{NLOS} links.   \\
\indent Moreover, in the existing literature described above, the effect of \ac{UE} location distibution and \ac{UE} density have not been explicitly captured in evaluation of downlink system performance. Whereas, one of the most distinctive feature of future generation \ac{UDN} has been identified as comparable \ac{AP} and \ac{UE} density \cite{Kamel_2016}. Such system configuration results in inhomogeneous inactivity of some \ac{AP}s whose effect are yet to be captured in \ac{SCN} modeled by 3-D \ac{SPPP}. \\
\indent In light of the existing state-of-the-art and gaps therein, the main contributions of this work can be summarized as follows:
\begin{itemize}
 \item A holistic system model representing a dense co-channel urban/indoor \ac{RAN} has been developed to include the realistic propagation effects of SSF (LOS/NLOS), two-slope distance-dependent PL, interferer activity closely following realistic channel models described in \cite{3gpp36814} where the locations of both \ac{AP}/\ac{UE} are modeled as 3-D \ac{SPPP}.
 \item The activity of \ac{AP}s are analytically modeled as a function of \ac{AP} and \ac{UE} density for such a 3-D network (Lemma \ref{lemma_activ_prob}). 
 \item Numerically simpler upper and lower bounds of downlink coverage probability have been evaluated considering all above mentioned effects (Theorem \ref{theo_cov_prob_deriv}). The upper bound is found to be extremely tight and thus can be used as a close approximation of coverage probability. 
\end{itemize}

\section{System Model and Assumptions}
\label{sec:sys_mod}
\subsection{Network Model}
%Describe system model, $\intppputh$ \\
%write the PPPs, lemma 1 (value of $p_a$)\\
%channel description. \\
%Approximation of LOS prob, lemma 2 (prob of LOS) \\
%%%%% Test of macros
%$\thres$, $\variance{x}$
%%%%%
The cellular network model with small cell \ac{AP}s are considered to be arranged according to a homogeneous \ac{SPPP}  $\Propppaputh$ with intensity $\intpppaputh$ defined in $\mathbb{R}^3$. The \ac{UE}s are also considered to be distributed according to another homogeneous \ac{SPPP} $\Propppue$ with intensity $\intpppue$ defined in $\mathbb{R}^3$. The processes $\Propppaputh$ and $\Propppue$ are considered to be mutually independent in this work, whereas the dependence of these processes and its effect can be explored in future. Furthermore, The \ac{UE}s are considered to be associated with closest \ac{AP}s \cite{Andrews_2011}, thus resulting in the coverage regions becoming a $3$-D Voronoi tessellation.  \\
\indent As in this work $\intpppue$ is considered to be comparable or even less than $\intpppaputh$ which is a distinguishable property of \ac{UDN} \cite{Kamel_2016}, some \ac{AP}s can be in inactive mode without any \ac{UE} under their coverage. Let the activity probability of an \ac{AP} is denoted by $\actprob$. The value of $\actprob$ depends on the intensity functions of the SPPPs representing the spatial distributions of \ac{AP} and \ac{UE}, as derived in Lemma \ref{lemma_activ_prob}. 
\begin{lemma}
\label{lemma_activ_prob}
For the \ac{AP}s distributed according to $\Propppaputh (\subset \mathbb{R}^3)$ with intensity $\intpppaputh$ serving \ac{UE}s distributed according to $\Propppue (\subset \mathbb{R}^3)$ with intensity $\intpppue$ with `closest AP association rule', the activity probability ($\actprob$) is given by: 
\begin{equation}
\label{eq:act_prob_exp}
 \actprob = 1 - \bigg[1+\frac{\intpppue}{5\intpppaputh}\bigg]^{-5}.
\end{equation}
\end{lemma}
The proof of the lemma is provided in Appendix \ref{Appendix_lemma1}. \\
\indent Let $\Propppapth (\subset \mathbb{R}^3)$ represents the point process for the active \ac{AP}s. Although theoretically $\Propppapth$ is not a homogeneous \ac{SPPP} due to the fact that the  points from the process are not obtained independently of each other \cite[Sec. 3.5.2]{Galiotto_2017}, it has been shown in \cite{Lee_2012} that the homogeneous \ac{SPPP} obtained by thinning of $\Propppaputh$ can efficiently approximate $\Propppapth$. Thus, using Thinning theorem \cite{Chiu_2013book}, the resulting \ac{SPPP} $\Propppapth$ has intensity $\intpppapth = \actprob\intpppaputh$.    
\subsection{Channel Model}
\label{subsec:channel_mod}
In this work we consider that the \ac{AP}s transmit with single isotropic antenna with transmit power $\transPow$ to \ac{UE}s with single isotropic antenna. Use of multiple and directional antennas at both \ac{AP}s and \ac{UE}s can be seen as a potential future work.  
In this work, the links between the \ac{UE} and serving as well as interfering \ac{AP}s can be \ac{LOS} or \ac{NLOS} depending on the following distance-dependent LOS probability model as proposed by 3GPP for picocell environment \cite[Table A.2.1.1.2-3]{3gpp36814}:
\begin{equation}
 \label{eq:los_prob_exact}
 \plosex{d} = 0.5-\text{min}(0.5, 5e^{-\frac{0.156}{d}})+\text{min}(0.5, 5e^{-\frac{d}{0.03}}).
\end{equation}
Clearly, the expression does not offer tractability for mathematical analysis. Therefore, several simple approximations of \eqref{eq:los_prob_exact} have been proposed in literature \cite[Sec. II.B]{Galiotto_2015}. In this work, we use the following approximation which simultaneously offer close match with \eqref{eq:los_prob_exact} along with mathematical tractability: 
\begin{equation}
 \label{eq:los_prob_approx}
 \plosap{d} = \text{exp}{((-\frac{d}{L})^3)}, ~~~~~~~~~~~~~~\text{with}~~~ L = 82.5.
\end{equation}
In this work we consider the following distance-dependent two-slope \ac{PL} model in order to represent both \ac{LOS} and \ac{NLOS} links:
\begin{equation}
  PL(d)=\begin{cases}
    \PLlos{d} = K_L d^{\PLElos}, & \text{with prob. $\plosex{d}$} ;\\
    \PLnlos{d} = K_{NL} d^{\PLEnlos}, & \text{with prob. $1-\plosex{d}$},
  \end{cases}
\end{equation}
where, $\PLElos$ and $\PLEnlos$ represent path loss exponents in \ac{LOS} and \ac{NLOS} links respectively; $K_L$ and $K_{NL}$ represent the path loss at unit distance ($d=1\text{m}$) for \ac{LOS} and \ac{NLOS} links respectively.\\
\indent We represent the small scale power fading terms in \ac{NLOS} and \ac{LOS} links by $\ssfnlos$ and $\ssflos$ respectively. We assume that in \ac{NLOS} links, propagation is affected by Rayleigh fading, making $\ssfnlos \sim \exp{(1)}$. In \ac{LOS} links, the small scale fading is characterized by nakagami-$m$ distribution, making $\ssflos$ follow normalized Gamma distribution with shape parameter $\losGamShp$ ($\ssflos\sim \Gamma (\losGamShp, \frac{1}{\losGamShp})$).   \\
\indent Without lack of generality we consider a typical user $u$ located at origin. Full frequency reuse has been considered among all the \ac{AP}s. It has been assumed that the \ac{UE} is attached with an \ac{AP} situated at $x_0$, with other \ac{AP}s at locations $x_i$ acting as interferers ($x_i \in \Propppapth \forall i$). Considering an interference-limited scenario applicable for dense networks, the received $\sir$ (denoted by $\Gamma_u$) experienced by the \ac{UE} can be written as:%obtained from \eqref{eq:tot_rcv_pwr}: 
% \begin{equation}
% \label{eq:SIR_exp}
%  \Gamma_u = \frac{\losindi{r_{x_0}}\ssflos\PLlosinv{r_{x_0}}+[1-\losindi{r_{x_0}}]\ssfnlos \PLnlosinv{r_{x_0}}}{\sum_{x_i\in \Propppapth\textbackslash\{x_0\}}\losindi{r_{x_i}}\ssflos\PLlosinv{r_{x_i}}+[1-\losindi{r_{x_i}}]\ssfnlos \PLnlosinv{r_{x_i}}}. 
% \end{equation}
\begin{align}
 \label{eq:SIR_exp}
 &\Gamma_u =  \\ \nonumber
 & \frac{\losindi{r_{x_0}}\ssflos\PLlosinv{r_{x_0}}+[1-\losindi{r_{x_0}}]\ssfnlos \PLnlosinv{r_{x_0}}}{\sum_{x_i\in \Propppapth\textbackslash\{x_0\}}\losindi{r_{x_i}}\ssflos\PLlosinv{r_{x_i}}+[1-\losindi{r_{x_i}}]\ssfnlos \PLnlosinv{r_{x_i}}}.
\end{align}
% \begin{align}
% \label{eq:SIR_exp}
%  &\Gamma_u = \\ \nonumber
%  &\frac{\losindi{r_{x_0}}\ssflos\PLlosinv{r_{x_0}}+[1-\losindi{r_{x_0}}]\ssfnlos \PLnlosinv{r_{x_0}}}{\sum_{x_i\in \Propppapth\textbackslash\{x_0\}}\losindi{r_{x_i}}\ssflos\PLlosinv{r_{x_i}}+[1-\losindi{r_{x_i}}]\ssfnlos \PLnlosinv{r_{x_i}}}. 
% \end{align}
% combined received power in downlink by that \ac{UE} from desired and interfering APs can be expressed as:
% \begin{align}
%  \label{eq:tot_rcv_pwr}
%  P_r &= \losindi{r_{x_0}}\ssflos\PLlosinv{r_{x_0}}+[1-\losindi{r_{x_0}}]\ssfnlos \PLnlosinv{r_{x_0}} \nonumber \\
%       &+ \sum_{i\in \Propppapth\textbackslash\{x_0\}}\losindi{r_{x_i}}\ssflos\PLlosinv{r_{x_i}}+[1-\losindi{r_{x_i}}]\ssfnlos \PLnlosinv{r_{x_i}},  
%  \end{align}
where $r_{x_i} \triangleq ||x_i|| \forall i$ and $\losindi{.}$ representing \ac{LOS} indicator Bernoulli \ac{RV} with following probability mass function: $\mathbb{P}(\losindi{d}=1)=\plosex{d}$ and $\mathbb{P}(\losindi{d}=0)=1-\plosex{d}$. 
\section{Analysis of Network Performance and Main Results}\label{sec:analysis}
%Main mathematical contributions
Let, from the target \ac{UE} $u$, the distance to the $n$-th nearest \ac{AP} is denoted by $r_n$. The distribution of $r_n$ follows the \ac{PDF} \cite[Sec. 3.2.2, Eq. (23)]{Moltchanov_2012}:
\begin{equation}
\label{eq:nth_AP_dist_pdf}
 \nAPdistpdf{n}{r_n} = \frac{3(\frac{4}{3}\pi\intpppapth r_n^3)}{r_n \Gamma(n)}\exp({-\frac{4}{3}\pi\intpppapth r_n^3)}, ~n=0, 1, 2 \cdots, r \geq 0.
\end{equation}
Putting $n=1$ in \eqref{eq:nth_AP_dist_pdf}, the \ac{PDF} of the distance from target \ac{UE} to the serving \ac{AP}, denoted by $r$ is given by:  
\begin{equation}
\label{eq:serv_dist_pdf}
 \servdistpdf{r} = 4\pi\intpppapth r^2\exp({-\frac{4}{3}\pi\intpppapth r^3)}, ~~~~~~~~~~r \geq 0.
\end{equation}
% In the following lemma the probability that the serving links experiences \ac{LOS} condition in the above described system has been evaluated.  
\begin{lemma}
\label{lemma_link_los_prob}
The probability that the link to the target \ac{UE} $u$ from the $n$-th nearest \ac{AP} belonging to $\Propppapth$ from it experiences \ac{LOS} condition can be approximated as: 
\begin{equation}
\label{eq:serv_link_los_prob}
 p^{n}_{LOS} \approx \big(\frac{4\pi\intpppapth L^3}{3+4\pi\intpppapth L^3}\big)^n.
\end{equation}
\end{lemma} 
\begin{proof}
%From \cite[Sec. 3.2.2, Eq. (23)]{Moltchanov_2012}, 
From the definition of \ac{LOS} probability made in \eqref{eq:los_prob_exact}, probability that the link from the \ac{UE} to $n$-th nearest \ac{AP} is \ac{LOS} given that its length $r_n$ equals $\plosex{r}$. Now, using the property of conditional probability, it can be written that: 
\begin{equation}
\label{eq:nth_AP_los_prob_exact}
 p^{n}_{LOS} = \int_0^\infty \plosex{r_n} \nAPdistpdf{n}{r} {\rm d}r_n% = \plosex{r}. 
\end{equation}
Using \eqref{eq:serv_dist_pdf} and invoking the approximation made in \eqref{eq:los_prob_approx}, \eqref{eq:nth_AP_los_prob_exact} can be rewritten as: 
\begin{align}
\label{eq:nth_AP_los_prob_approx}
 p^{n}_{LOS} &\approx \int_0^\infty \text{exp}{((-\frac{d}{L})^3)} \frac{3(\frac{4}{3}\pi\intpppapth r_n^3)}{r_n \Gamma(n)}\exp({-\frac{4}{3}\pi\intpppapth r_n^3)} {\rm d}r_n\\ \nonumber
 &= \frac{3(\frac{4}{3}\pi\intpppapth)^n}{(n-1)!}\int_0 ^\infty r_n^{3n-1}\text{exp}(-(\frac{4}{3}\pi\intpppapth+\frac{1}{L^3})r_n^3){\rm d}r_n.
\end{align}
Using the substitution: $v = 3r_n^2(\frac{4}{3}\pi\intpppapth+\frac{1}{L^3})$, \eqref{eq:nth_AP_los_prob_approx} can be rewritten as: 
\begin{align}
\label{eq:nth_AP_los_prob_substitution}
 p^{n}_{LOS} &= \frac{(\frac{4}{3}\pi\intpppapth)^n}{(n-1)!(\frac{4}{3}\pi\intpppapth+\frac{1}{L^3})}\int_0^\infty \big[\frac{v}{\frac{4}{3}\pi\intpppapth+\frac{1}{L^3}}\big]^{n-1}\text{exp}(-v){\rm d}v\\ \nonumber
 &= \frac{(\frac{4}{3}\pi\intpppapth)^n}{(n-1)!(\frac{4}{3}\pi\intpppapth+\frac{1}{L^3})^n}\underbrace{\int_0^\infty v^{n-1}\text{exp}(-v){\rm d}v }_\text{$=\Gamma(n)=(n-1)!$}.
\end{align}
Simple algebraic manipulations of \eqref{eq:nth_AP_los_prob_substitution} finally results in \eqref{eq:serv_link_los_prob}.
\end{proof}
\begin{remark}
The probability that the link between the target \ac{UE} and its serving \ac{AP} is given by: 
$p^{n}_{LOS} \approx \big(\frac{4\pi\intpppapth L^3}{3+4\pi\intpppapth L^3}\big)$. Putting $n=1$ in Lemma \ref{lemma_link_los_prob} gives this result.
\end{remark}
\subsection{Bounds of Coverage probability}
\indent The coverage probability ($\covprob$) experienced by \ac{UE} $u$ can be defined as: $\covprob \triangleq \mathbb{P}(\Gamma_u\geq \thres)$, where $\thres$ is a predefined $\sir$ threshold. Using the $\sir$ expression given in \eqref{eq:SIR_exp}, the average coverage probability experienced by the particular \ac{UE} $u$ is given by: 
\begin{align}
 \label{eq:covprob_deriv}
& \covprob = \Exptop{\mathbb{P}(\Gamma_u\geq \thres|r)}{r} = \int_{0}^{\infty} \mathbb{P}(\Gamma_u\geq \thres|r)\servdistpdf{r}{\rm d}r \\ \nonumber
 &= \int_0^{\infty} 4\pi\intpppapth r^2\exp({-\frac{4}{3}\pi\intpppapth r^3)} \big[\plosap{r}\mathbb{P}(\ssflos>\thres \interf K_{L}r^{\PLElos}|r)  \\ \nonumber
&~~~~~~~+(1-\plosap{r})\mathbb{P}(\ssfnlos>\thres \interf K_{NL}r^{\PLEnlos}|r) \big],
\end{align}
where $\interf$ denotes the aggregate interference power received at the \ac{UE}. 
Using the fact that $\ssfnlos\sim \exp(1)$, it can be written that
\begin{equation}
\label{eq:nlos_ssf_ccdf}
\mathbb{P}(\ssfnlos>\thres \interf K_{NL}r^{\PLEnlos}|r) = \interfLaplace{\thres \interf K_{NL}r^{\PLEnlos}},
\end{equation}
where $\interfLaplace{.}$ denotes the Laplace transform of $I$ and is defined as $\interfLaplace{s} \triangleq \Exptop{e^{-sI}}{I}$.    
\begin{lemma}
 \label{lemma_laplace_transform}
 In a 3-D network described in Sec. \ref{sec:sys_mod} in presence of both \ac{LOS} and \ac{NLOS} propagation, the Laplace transform of aggregate interference is given by: 
  \scriptsize{
 \begin{align}
  \label{eq:laplace_transform_expr}
  &\interfLaplace{s} = \exp\big[ -4\pi\intpppapth\int_{r}^\infty \big(1-\frac{1}{1+sK_{NL}^{-1}x^{-\PLEnlos}}\big)x^2{\rm d}x\big] \times \\ \nonumber
 & \exp\big[ -4\pi\intpppapth\int_{r}^\infty \plosap{x}\big(\frac{1}{1+sK_{NL}^{-1}x^{-\PLEnlos}}-\frac{1}{(1+\frac{sK_{L}^{-1}x^{-\PLElos}}{\losGamShp})^{\losGamShp}}\big)x^2{\rm d}x\big]
 \end{align}}
\end{lemma}
 \normalfont
The proof of \eqref{eq:laplace_transform_expr} has been provided in Appendix \ref{Appendix_lemma3}.\\
As $\ssflos\sim \Gamma(\losGamShp, \frac{1}{\losGamShp})$, according to the derivation provided in \cite[Theorem 1]{Atzeni_2018}, it can be written that: 
\begin{equation}
 \label{eq:los_ssf_ccdf}
\mathbb{P}(\ssflos>\thres \interf K_{L}r^{\PLElos}|r) = \sum_{k=0}^{\losGamShp-1} \frac{(-s)^k}{k!} \kthordrdiff \interfLaplace{s}\bigg|_{s=\losGamShp\thres K_L r^{\PLElos}}. 
\end{equation}
Applying \eqref{eq:nlos_ssf_ccdf} and \eqref{eq:los_ssf_ccdf} in \eqref{eq:covprob_deriv} with the aid of \eqref{eq:laplace_transform_expr}, it can be clearly understood that the exact expression of $\covprob$ involves successive differentiation of $\interfLaplace{s}$, upto the order of $\losGamShp-1$, making it numerically intractable and cumbersome to evaluate. Therefore, in the following theorem, numerically simpler upper and lower bounds of $\covprob$ have been derived.
\begin{theorem}
 \label{theo_cov_prob_deriv}
 The coverage probability ($\covprob$) of atypical \ac{UE} in the downlink cellular network described in Sec \ref{sec:sys_mod} can be bounded as shown in \eqref{eq:cov_prob_bounds}, where $\zeta' = \losGamShp(\losGamShp!)^{-\frac{1}{\losGamShp}}$. 
\begin{figure*}
%\tiny 
\small
 \begin{align}
  \label{eq:cov_prob_bounds}
  \int_0^\infty \bigg[ (1-\plosap{r})\interfLaplace{\thres K_{NL}r^{\PLEnlos}}&+\plosap{r}\sum_{l=1}^\losGamShp (-1)^{l+1}{\losGamShp \choose l} \interfLaplace{\losGamShp\thres K_L r^{\PLElos}l}\bigg] \servdistpdf{r} {\rm d} r ~~<~~ \covprob ~~<  \nonumber \\
  &\int_0^\infty \bigg[ (1-\plosap{r})\interfLaplace{\thres K_{NL}r^{\PLEnlos}}+\plosap{r}\sum_{l=1}^\losGamShp (-1)^{l+1} {\losGamShp \choose l} \interfLaplace{\zeta'\thres K_L r^{\PLElos}l}\bigg] \servdistpdf{r} {\rm d} r
 \end{align}
 \rule{0.97\textwidth}{0.5pt}\vspace{-0.7cm}
 \end{figure*}
 \end{theorem}
\begin{proof}
For a normalized Gamma distributed random variable $X$ with shape parameter $\losGamShp $, its \ac{CDF} can be formulated as: 
\begin{equation}
 \label{eq:norm_gamma_dist_cdf}
 \mathbb{P}(X\leq x) = \frac{1}{\Gamma(\losGamShp)}\lincGammafunc{\losGamShp}{\losGamShp x},
\end{equation}
where $\lincGammafunc{.}{.}$ is the lower incomplete Gamma function defined as: 
\begin{equation}
 \label{eq:linc_gamma_func_defn}
 \lincGammafunc{s}{x} = \int_0^x t^{s-1}e^{-t}{\rm d}t. 
\end{equation}
From \cite[Theo. 1]{Alzer_1997} and using the fact that $\int_0^z e^{-t^p}{\rm d}t = \frac{1}{p}\lincGammafunc{\frac{1}{p}}{z^p}$, following inequality can be written: 
\begin{equation}
 \label{eq:alzer_ineq}
 [1-e^{-\zeta z^p}]^{1/p} < \frac{1}{p\Gamma(1+\frac{1}{p})}\lincGammafunc{\frac{1}{p}}{z^p} < [1-e^{-\alpha z^p}]^{1/p},
\end{equation}
where $\alpha = 1$ and $\zeta = [\Gamma(1+1/p)]^{-p}$ if $0<p<1$. Using the substitution $p = 1/\losGamShp $ (which ensures $p<1$, as for all practical scenarios, $\losGamShp >1$), \eqref{eq:alzer_ineq} can be rewritten as: 
\begin{equation}
 \label{eq:alzer_ineq_reformulated}
 [1-e^{-\zeta z^{1/\losGamShp }}]^{\losGamShp } < \frac{1}{\Gamma(\losGamShp )}\lincGammafunc{\losGamShp }{z^{1/\losGamShp }} < [1-e^{-\alpha z^{1/\losGamShp }}]^{\losGamShp }
\end{equation}
Finally, using the substitution $z^{1/\losGamShp } = \losGamShp x$ and using \eqref{eq:norm_gamma_dist_cdf}, \eqref{eq:alzer_ineq_reformulated} can be rewritten as: 
\begin{align}
 \label{eq:norm_gamma_cdf_ineq}
 [1-e^{-\zeta \losGamShp x}]^{\losGamShp } &< \mathbb{P}(X\leq x) < [1-e^{-\alpha \losGamShp x}]^{\losGamShp }, \nonumber \\
 \implies 1-[1-e^{- \losGamShp x}]^{\losGamShp } &< \mathbb{P}(X\geq x) < 1-[1-e^{-\zeta \losGamShp x}]^{\losGamShp}
 \end{align}
Putting \eqref{eq:norm_gamma_cdf_ineq} in \eqref{eq:covprob_deriv} due to the fact that $\ssflos$ is a normalized Gamma function with shape parameter $\losGamShp$ and using Binomial theorem, \eqref{eq:cov_prob_bounds} can be obtained.
\end{proof}
In the later sections, it has been shown that the upper bound provided in \eqref{eq:cov_prob_bounds} is extremely tight and thus can be used as a close approximate of $\covprob$. %Therefore, in the subsequent sections it has been used to represent $\covprob$ in order to formulate important system level performance metrics.  
\section{Numerical Results and Discussions}
%$\losGamShp \choose l$
%Table of system parameters\\
\begin{table}[h]
\caption{\small{Key system parameters}}
\centering
%\label{table_SystemParameter}
\begin{tabular}{|m{4.5cm}|m{3cm}|}\hline
%\begin{tabular}{|l|l|}\hline
\textbf{Parameter} & \textbf{Value}\\ \hline
% \multirow{2}{*}{Network topology} & Indoor office \cite[Fig. 1]{maccartney_2015}  \\ \cline{2-2}
% &  Hexagonal cells with ISD = $500m$ for UMa \cite{3gpp36873}. \\ \hline
%Carrier Frequency (for sub 6 GHz) & 2.4 GHz  \\ \hline
Intensity of $\Propppaputh$ ($\intpppaputh$, in $m^{-3}$) & $10^{-6}-10^{2}$  \\ \hline
Intensity of $\Propppue$ ($\intpppue$, in $m^{-3}$) & $10^{-4}, 10^{-2}, 10^{0}$  \\ \hline
Carrier frequency & 2 GHz  \\ \hline
\multirow{2}{*}{\ac{PL} parameters in LOS link \cite{3gpp36814}} & $K_L = 10^{4.11}$ \\ \cline{2-2}
&  $\PLElos= 2.09$ \\ \hline
\multirow{2}{*}{\ac{PL} parameters in NLOS link \cite{3gpp36814}} & $K_{NL} = 10^{3.29}$ \\ \cline{2-2}
&  $\PLEnlos= 3.75$ \\ \hline
Shape factor of Gamma distribution in LOS links ($\losGamShp$) & $3$ (assumed)  \\ \hline
$\sir$ threshold for coverage ($\thres$) & $-10$ dB (assumed)  \\ \hline
\end{tabular}
%\caption{Key system parameters}
\label{table_SystemParameter}
%\vspace{-0.75cm}
\end{table}
In this section we show the accuracy of the analytical results obtained in Sec. \ref{sec:analysis} by comparing them with the ones obtained from Monte-Carlo simulation. The path loss parameters for \ac{LOS} and \ac{NLOS} links are chosen by following realistic channel models prescribed in \cite[Table A.2.1.1.2-3]{3gpp36814} and listed in Table \ref{table_SystemParameter} along with other key system parameters. It is worth noting that for LOS links, no typical value of $\losGamShp$ has been provided in measurement campaigns. In this work we consider $\losGamShp$ = 3, although the analytical framework is generic to accommodate any other value of $\losGamShp$ as well. The $\sir$ threshold for coverage ($\thres$) is assumed to be $-10$ dB, however the work is generic to accommodate any other value. 
\subsection{Validation of Lemma \ref{lemma_activ_prob}}
\begin{figure}[h]
   \centering
\includegraphics[width= \linewidth]{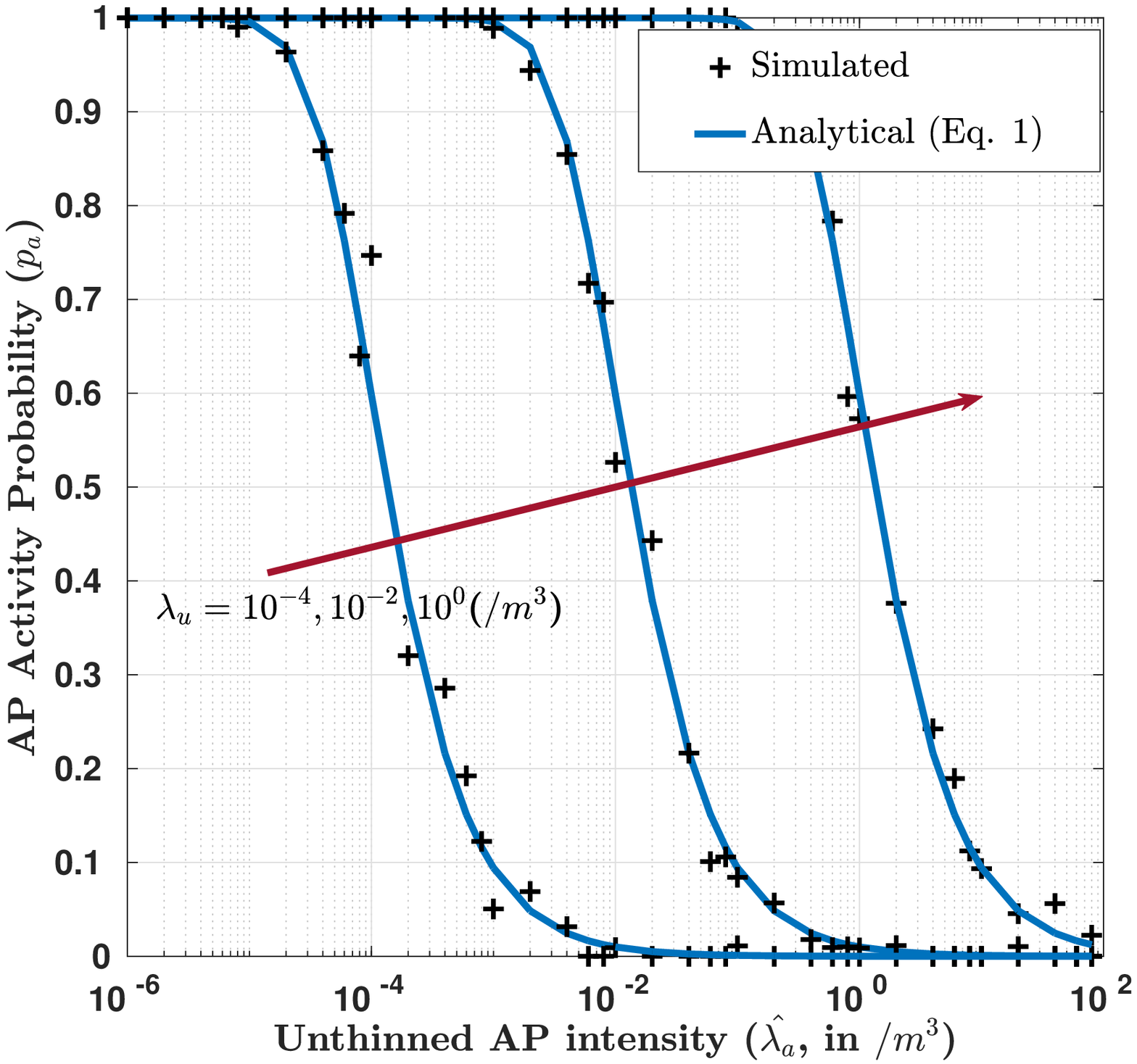}
   \caption{\small{Variation of \ac{AP} activity probability ($\actprob$) w.r.t. varied AP intensity ($\intpppaputh$) for various UE intensity ($\intpppue$).}}
   \label{fig:lemma_1}
\end{figure}
In Fig \ref{fig:lemma_1} the accuracy of Lemma \ref{lemma_activ_prob} has been shown in comparison with results obtained from Monte-Carlo simulation for a wide range of $\intpppaputh$ and $\intpppue$. The obtained closed match between the simulated and analytically obtained curves shows the usability of Eq. \eqref{eq:act_prob_exp} for further analysis. The nature of activity probability ($\actprob$) curves are intuitive in nature. For a given \ac{UE} intensity ($\intpppue$), with increase in \ac{AP} intensity ($\intpppaputh$), $\actprob$ decreases, due to the fact that for each \ac{UE}, number of candidate \ac{AP}s increases, therefore resulting in inactivity of some \ac{AP}s. On the other hand, for a given $\intpppaputh$, with increase in $\intpppue$ (from $10^{-4}$ to $10^0/m^{3}$), $\actprob$ increases due to the fact that in order to cater to increasing service deman, more and more \ac{AP}s need to be kept in `active' state.
\subsection{Validation of Lemma \ref{lemma_link_los_prob}}
\begin{figure}[h]
   \centering
   \includegraphics[width= \linewidth]{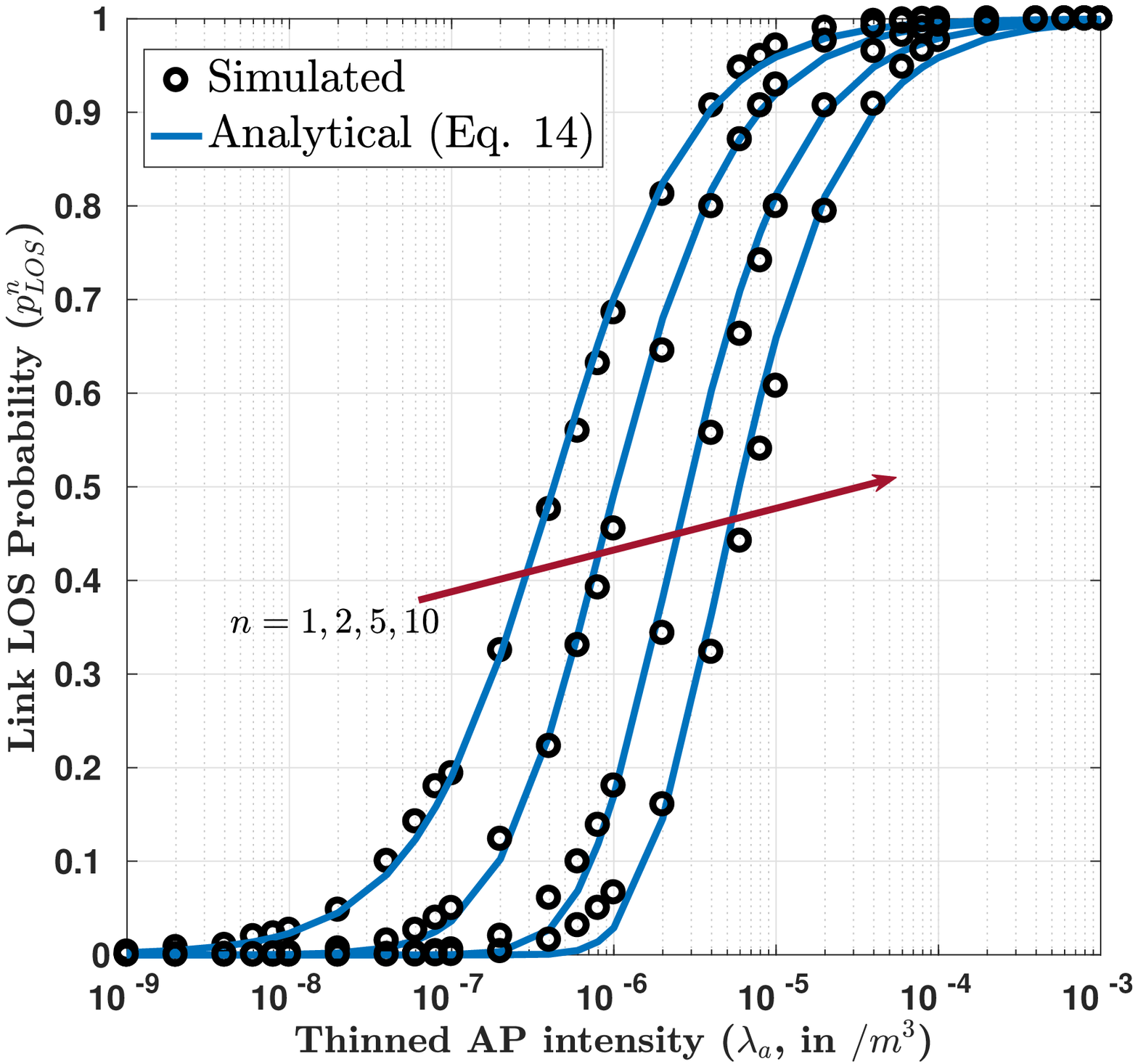}
   \caption{\small{Variation of link \ac{LOS} probability ($p^n_{LOS}$) w.r.t varies $\intpppapth$ for different $n$.}}
   \label{fig:lemma_2}
\end{figure}
The analytical correctness of Lemma \ref{lemma_link_los_prob} has been exhibited in Fig. \ref{fig:lemma_2}, where the numerical values of $p^{n}_{LOS}$ have been compared to the ones obtained from simulation for a wide range of $\intpppapth$ and $n$. The observed close match between the siulated and analytical results proves the correctness of Eq. \eqref{eq:serv_link_los_prob}. By analyzing the \ac{LOS} probability curves, it can be seen that for any specific $n$, with increase in $\intpppapth$, probability that the link becomes \ac{LOS} increases. This is primarily due to the fact that with more densification of \ac{AP}s, expected length of such links reduces, resulting in increase in $p^{n}_{LOS}$. On the other hand for a specific $\intpppapth$, $p^{n}_{LOS}$ decreases with increase in $n$ due to the fact that the expected length increases with increase in $n$, resulting in decrease in LOS probability. Note that in Fig. \ref{fig:lemma_2}, the curve of $p^{n}_{LOS}$ has been shown for an extended range of $\intpppapth$ ($10^{-9}$-$10^{-1}/m^{3}$). This is due to show the accuracy of \eqref{eq:serv_link_los_prob} for the whole range of $p^n_{LOS}$ (0-1). 
\subsection{Validation of Theorem \ref{theo_cov_prob_deriv} and Main Result}
\begin{figure}[h]
   \centering
   \includegraphics[width= \linewidth]{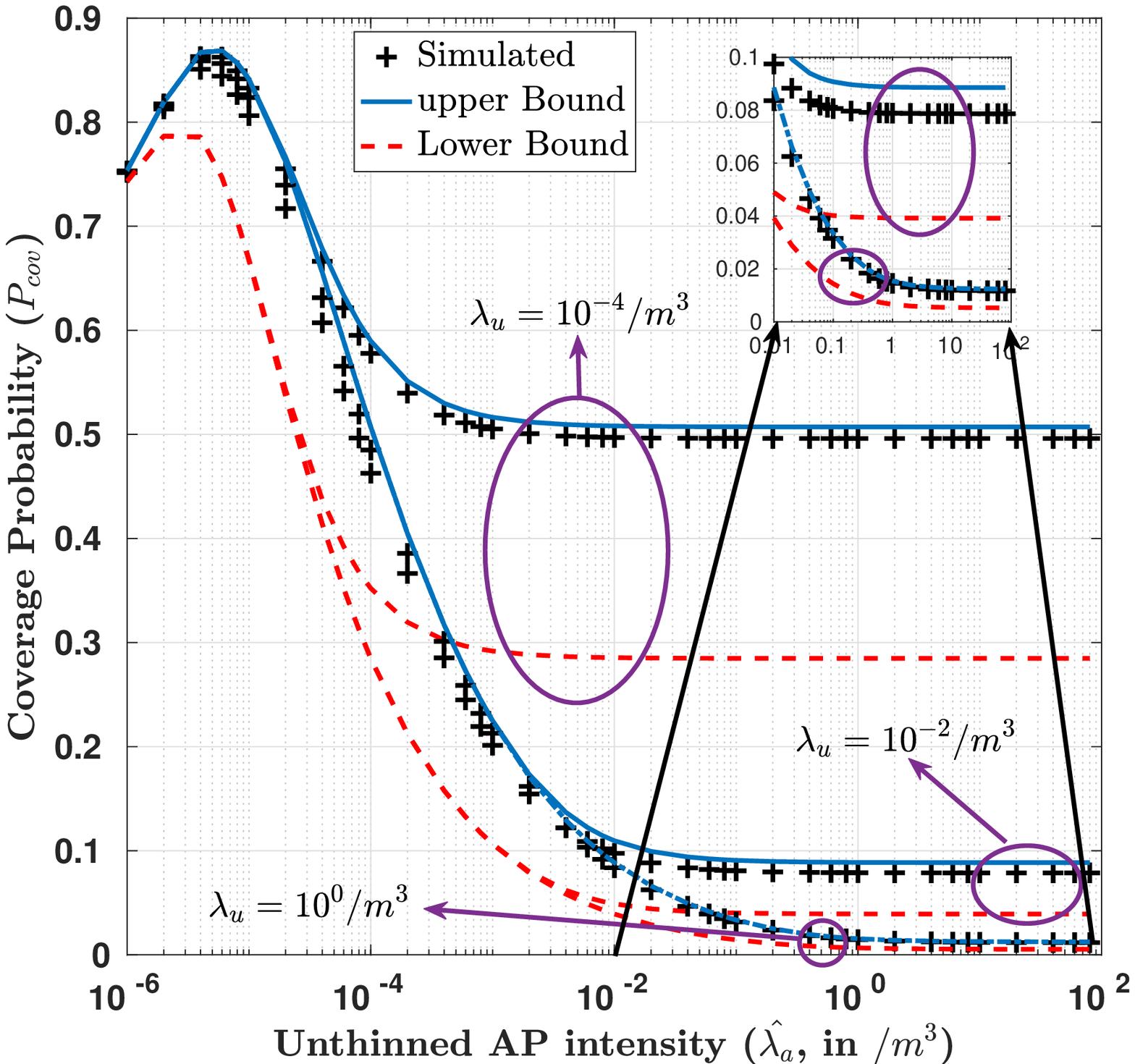}
   \caption{\small{Simulated and analytical bounds of coverage probability ($\covprob$), as in \ref{eq:cov_prob_bounds}) w.r.t. \ac{AP} intensity ($\intpppaputh$) for various \ac{UE} intensity ($\intpppue$).}}
   \label{fig:theo_1}
\end{figure}
The analytical acceptance of the derived upper and lower bounds of downlink coverage probability obtained in Theorem \ref{theo_cov_prob_deriv} has been shown in Fig. \ref{fig:theo_1} by comparing them with results obtained from Monte-Carlo simulation for a wide range of \ac{AP} intensity ($\intpppaputh$) for low, moderate and high \ac{UE} intensity ($\intpppue$). By observing the result, it can be inferred that the theoretical upper bound is extremely tight in nature for the whole range of system parameters. However, the derived lower bound is loose in nature. \\
\indent The plotted curves of $\covprob$ also reveal important impact of \ac{LOS}/\ac{NLOS} propagation. For the lower range of $\intpppaputh$, we observe an increasing trend of $\covprob$ until it reaches a maximum level. In this range of $\intpppaputh$, the link between the \ac{UE} and the serving (nearest) \ac{AP} is in \ac{LOS} condition. Whereas, most of the interferer links are in \ac{NLOS} condition. After reaching the maximum limit, the coverage probability decreases monotonically with further increase in $\covprob$, due to the fact that the interfering \ac{AP}s also start to enter \ac{LOS} zone, thus increasing the aggregate interference power at a higher rate than the rate at which the desired power increases. Finally, at further higher range of $\intpppaputh$, the coverage probability remains broadly unchanged, due to the fact that in such scenarios, all the desired and interfering \ac{AP}s reside in \ac{LOS} zone, and adding more and more \ac{AP}s does not affect much in the resulting $\sir$.
It has also been observed that with increase in $\intpppue$ (from $10^{-4}$ to $10^0$), coverage probability decreases, due to the fact as \ac{UE} density increases, the \ac{AP} activity probability increases, resulting in increase in number of interfering \ac{AP} and subsequently aggregate interference power, which finally reduces $\sir$. 
\section{Conclusions}
\label{sec:conclusions}
In this work, analytical upper and lower bounds of downlink coverage probability experienced by a typical \ac{UE} in a 3-D heterogeneous \ac{SCN} described by homogeneous \ac{SPPP}. In this analysis, the holistic effects of \ac{LOS}/\ac{NLOS} propagation in terms of realistic link \ac{LOS} probability, two-slope \ac{PL} along with \ac{SSF} have been considered. The derived upper bound is shown to be numerically simpler as well as extremely tight in nature compared to the exact simulated coverage probability and thus can be used as a close approximate of the same. Thus, the presented result can further be used in solving emerging \ac{RAN} design issues like area spectral/energy efficiency evaluation/maximization, \ac{AP} density optimization etc.     
\appendices
\begin{figure*}
%\tiny 
\small
\begin{align}
\label{eq:laplace_transform_deriv}
 \interfLaplace{s} &= \Exptop{e^{-sI}}{I} = \Exptop{\exp\big(-s\sum_{x_i\in\Propppapth\textbackslash\{x_0\}}
 \losindi{r_{x_i}}\ssflos\PLlosinv{r_{x_i}}+[1-\losindi{r_{x_i}}]\ssfnlos \PLnlosinv{r_{x_i}}\big)}{(\Propppapth, \losindi{.},\ssflos, \ssfnlos)}  \nonumber \\
 &= \Exptop{\prod_{x_i\in\Propppapth\textbackslash\{x_0\}}\big(\plosap{r_{x_i}} \Exptop{\exp\big(-s \ssflos K_L^{-1} r_{x_i}^{-\PLElos}\big)}{\ssflos}+(1-\plosap{r_{x_i}}) \Exptop{\exp\big(-s \ssfnlos K_{NL}^{-1} r_{x_i}^{-\PLEnlos}\big)}{\ssfnlos}}{\Propppapth}  \nonumber \\ 
 &\stackrel{(a)}{=} \exp(-4\pi\intpppapth\int_r^{\infty}[1-\Exptop{\exp(-s\ssfnlos K_{NL}^{-1}x^{-\PLEnlos})}{\ssfnlos} x^2 {\rm d}x])  \nonumber \\
 &\times \exp(-4\pi\intpppapth\int_r^{\infty}\plosap{x}[(\Exptop{\exp(-s\ssfnlos K_{NL}^{-1}x^{-\PLEnlos})}{\ssfnlos} - \Exptop{\exp(-s\ssflos K_{L}^{-1}x^{-\PLElos})}{\ssflos}) x^2 {\rm d}x])
\end{align}
\rule{0.97\textwidth}{0.5pt}
\end{figure*}
\section{Proof of Lemma \ref{lemma_activ_prob}} \label{Appendix_lemma1}
For \ac{SPPP} $\Propppaputh$, the Voronoi cells' volume ($v$) distribution can be approximated as \cite[Eq. 11]{Ferenc_2007}
\begin{equation}
\label{voronoi_cell_volume_dist}
 f_{3D}(v) = \frac{5^5}{\Gammafunc{5}}\intpppaputh^5v^4\text{exp}(-5v\intpppaputh). 
\end{equation}
From the property of \ac{UE} SPPP $\Propppue$, no \ac{UE} will lie in a Voronoi cell of volume $v$ with probability $\text{exp}(-v\intpppue)$. So, the average void probability in 3D SPPP ($\voidprob$) is given by: 
 \begin{align}
  &\voidprob = \int_{0}^{\infty} \text{exp}(-v\intpppue)f_{3D}(v) {\rm d} v   \\ \nonumber
	    &= \frac{(5\intpppaputh)^5}{\Gammafunc{5}} \int_{0}^{\infty} v^4 \text{exp}(-v(5\intpppaputh+\intpppue))=\bigg[1+\frac{\intpppue}{5\intpppaputh}\bigg]^{-5}.%\\ \nonumber
	  %  &= \bigg[1+\frac{\intpppue}{5\intpppaputh}\bigg]^{-5}.
  \label{voronoi_void_prob_eval}
 \end{align}
Thus, the activity probability ($\actprob$) of \ac{AP}s (i.e. probability that at least one \ac{UE} will lie in the Voronoi cell and thus be under its coverage) can be written as: 
\begin{equation}
  \actprob = 1-\voidprob = 1 - \bigg[1+\frac{\intpppue}{5\intpppaputh}\bigg]^{-5}.
  \label{activity_prob_eval}
 \end{equation}
 \section{Proof of Lemma \ref{lemma_laplace_transform}} \label{Appendix_lemma3}
% \begin{figure*}
% %\tiny 
% \small
% \begin{align}
% \label{eq:laplace_transform_deriv}
%  \interfLaplace{s} &= \Exptop{e^{-sI}}{I} = \Exptop{\exp\big(-s\sum_{x_i\in\Propppapth\textbackslash\{x_0\}}
%  \losindi{r_{x_i}}\ssflos\PLlosinv{r_{x_i}}+[1-\losindi{r_{x_i}}]\ssfnlos \PLnlosinv{r_{x_i}}\big)}{(\Propppapth, \losindi{.},\ssflos, \ssfnlos)}  \nonumber \\
%  &= \Exptop{\prod_{x_i\in\Propppapth\textbackslash\{x_0\}}\big(\plosap{r_{x_i}} \Exptop{\exp\big(-s \ssflos K_L^{-1} r_{x_i}^{-\PLElos}\big)}{\ssflos}+(1-\plosap{r_{x_i}}) \Exptop{\exp\big(-s \ssfnlos K_{NL}^{-1} r_{x_i}^{-\PLEnlos}\big)}{\ssflos}}{\Propppapth}  \nonumber \\ 
%  &\stackrel{(a)}{=} \exp(-4\pi\intpppapth\int_r^{\infty}[1-\Exptop{\exp(-s\ssfnlos K_{NL}^{-1}x^{-\PLEnlos})}{\ssfnlos} x^2 {\rm d}x])  \nonumber \\
%  &\times \exp(-4\pi\intpppapth\int_r^{\infty}\plosap{x}[(\Exptop{\exp(-s\ssfnlos K_{NL}^{-1}x^{-\PLEnlos})}{\ssfnlos} - \Exptop{\exp(-s\ssflos K_{L}^{-1}x^{-\PLElos})}{\ssflos}) x^2 {\rm d}x])
% \end{align}
% \rule{0.97\textwidth}{0.5pt}
% \end{figure*}
From the definition of Laplace transform of $I$, $\interfLaplace{s}$ can be expressed as \eqref{eq:laplace_transform_deriv} at the top of the next page, where step (a) follows from the \ac{PGFL} of 3-D \ac{SPPP} \cite{Chiu_2013book, Pan_2015, Omri_2016}, which states that for some function $f(x)$, $\Exptop{\Pi_{x\in \Propppapth\textbackslash\{x_0\}}f(x)}{\Propppapth} = \exp \big( -4\pi\intpppapth\int_{r}^{\infty}[1-f(x)]x^2{\rm d}x\big)$ and rearranging the terms.
Recalling that $\ssfnlos\sim \exp(1)$ and $\ssflos\sim \Gamma(\losGamShp, \frac{1}{\losGamShp})$, it can be written that:
%\begin{subequations}
 \begin{align}
 \label{eq:los_nlos_ssf_mgf}
  &\Exptop{\exp(-s\ssfnlos K_{NL}^{-1}x^{-\PLEnlos})}{\ssfnlos} = \frac{1}{1+sK_{NL}^{-1}x^{-\PLEnlos}}, \nonumber \\ 
 &\Exptop{\exp(-s\ssflos K_{L}^{-1}x^{-\PLElos})}{\ssflos} = \frac{1}{(1+\frac{sK_{L}^{-1}x^{-\PLElos}}{\losGamShp})^\losGamShp}.
\end{align} 
%\end{subequations}
Putting \eqref{eq:los_nlos_ssf_mgf} in \eqref{eq:laplace_transform_deriv}, \eqref{eq:laplace_transform_expr} can be readily obtained. 
\bibliographystyle{IEEEtran}
\bibliography{ref_list_udn_ppp} 
\end{document}